\documentclass[12pt,a4paper]{article}
\usepackage{epsf, graphicx}
\usepackage{latexsym,amsfonts,amsbsy,amssymb}
\usepackage{amsmath,amsthm}
\usepackage{cleveref}
\usepackage{indentfirst}
\usepackage{bm}
\usepackage [latin1]{inputenc}
\makeatletter

\@addtoreset{equation}{section} \makeatother
\newtheorem{Theorem}{Theorem}[section]
\newtheorem{Lemma}{Lemma}[section]

\newtheorem{Remark}{Remark}[section]
\newtheorem{Definition}{Definition}[section]
\newtheorem{Proposition}{Proposition}[section]
\newtheorem{Example}{Example}[section]
\makeatletter \@addtoreset{equation}{section} \makeatother
\begin{document}
\title{A class of constacyclic codes are generalized Reed-Solomon codes }

\author{Hongwei Liu,~Shengwei Liu\textsuperscript{*}}
\date{School of Mathematics and Statistics,
 Central China Normal University, Wuhan 430079, China}
\maketitle
\insert\footins{\small{\it Email addresses}: ~
 hwliu@ccnu.edu.cn (Hongwei Liu); ~
 shengweiliu@mails.ccnu.edu.cn (Corresponding author: Shengwei Liu).\\
 }
{\centering\section*{Abstract}}
 \addcontentsline{toc}{section}{\protect Abstract} 
 \setcounter{equation}{0} 
Maximum distance separable (MDS) codes are optimal in the sense that the minimum distance cannot be improved for a given length and code size. The most prominent MDS codes are generalized Reed-Solomon (GRS) codes. The square $\mathcal{C}^{2}$ of a linear code $\mathcal{C}$ is the linear code spanned by the component-wise products of every pair of codewords in $\mathcal{C}$. For an MDS code $\mathcal{C}$, it is convenient to determine whether $\mathcal{C}$ is a GRS code by determining the dimension of $\mathcal{C}^{2}$. In this paper, we investigate under what conditions that MDS constacyclic codes are GRS. For this purpose, we first study the square of constacyclic codes. Then, we give a sufficient condition that a constacyclic code is GRS. In particular, We provide a necessary and sufficient condition that a constacyclic code of a prime length is GRS.

\medskip
\noindent{\large\bf Keywords:~}\medskip  Constacyclic codes, GRS codes, MDS codes, Schur square

\noindent{\bf2010 Mathematics Subject Classification}: 94B05, 94B65.

\section{Introduction}
A linear code $\mathcal{C}$ of length $n$, dimension $k$ and minimum Hamming distance $d$ over the finite field $\mathbb{F}_{q}$ is called an $[\,n, k, d\,]_{q}$ code. If the parameters of the code $\mathcal{C}$ reach the Singleton bound $d=n-k+1$, then $\mathcal{C}$ is called a {\it maximum distance separable (MDS) code} \cite{R. Roth}. Generalized Reed-Solomon codes are one of the famous classes of MDS codes.
For other classes of MDS codes, we refer to \cite{A. Lempel} and \cite{G. Seroussi}. Given vectors $\bm x=(x_{1},\dots,x_{n}),~\bm y=(y_{1},\dots,y_{n})$ of $\mathbb{F}_{q}^{n}$, the Schur product of $\bm x$ and $\bm y$ is $\bm x*\bm y=(x_{1}y_{1},\dots,x_{n}y_{n})$. For two linear codes $C,~D\subseteq\mathbb{F}_{q}^{n}$, the {\it product code} $C*D$ is a linear subspace of $\mathbb{F}_{q}^{n}$ generated by all products $\bm x*\bm y$ where $\bm x\in C,~\bm y\in D$. When $C=D$, we call $C*C=C^{2}$ the {\it square} of $C$. For the basic results about product of linear codes, we refer to \cite{H. Randriambololona}. In \cite{I. Cascudo}, Cascudo characterized the square of cyclic codes by defining sets. In \cite{B.H. Falk}, the authors studied the power of constacyclic codes by pattern polynomials.
 In \cite{D. Mirandola}, Mirandola and Z\'{e}mor characterized the codes $C$ and $D$ whose product has maximum possible minimum distance and showed that an $[\,n, k, d\,]$ MDS code with $k\leq\frac{n-1}{2}$ is GRS if and only if $\dim(\mathcal{C}^{2})=2k-1$.
 In recent years, many authors applied the square of an MDS code to determine whether it is a GRS code (see \cite{P. Beelen},~\cite{H.LIU} and \cite{C.ZHU}). The key point is to determine the dimension of $\mathcal{C}^{2}$.

In \cite{A. Ashikhmin}, the authors showed that quantum MDS codes can be obtained by constructing Hermitian self-orthogonal MDS codes. In \cite{M. Grassl}, Grassl and R\"{o}tteler obtained new quantum MDS codes by shortening quantum MDS codes of length $q^{2}+1$ using puncture codes (see \cite{E. M. Rains}) and came up with their Conjecture 11 for a class of constacyclic codes. In \cite{S. BallR. Vilar}, Ball and Vilar determined when a truncated generalized Reed-Solomon code is Hermitian self-orthogonal, and proved the conjecture 11 of \cite{M. Grassl} for the version of GRS codes. In \cite{S. Ball}, Ball proved that the MDS constacyclic codes which given in \cite{M. Grassl} were GRS codes by a technical method. Combining with the results from \cite{S. BallR. Vilar}, Ball showed that Conjecture 11 from \cite{M. Grassl} is true.
On the other hand, in order to construct quantum MDS codes, many authors constructed Hermitian self-orthogonal MDS codes via constacyclic codes (for example \cite{X. Kai}, \cite{T. Zhang} and \cite{B. Chen}) and GRS codes (for example \cite{G. Guo}, \cite{G. Wang}). It is natural to consider whether these MDS constacyclic codes are GRS codes. Therefore, it is interesting to investigate under what conditions MDS constacyclic codes are GRS.

In this paper, we investigate under what conditions that a MDS constacyclic code $\mathcal{C}$ is GRS by determining the dimension of $\mathcal{C}^{2}$. For this purpose, we first characterize the square of a constacyclic code by it's defining set. Our result of square of a constacyclic code slightly generalizes the result in \cite{I. Cascudo}. Then, we give a sufficient condition that a constacyclic code is GRS. In particular, we provide a necessary and sufficient condition that a constacyclic code of a prime length is GRS. This paper is organized as follows. Section $2$ gives some necessary definitions and notions. In Section $3$, we study the square of constacyclic codes. In Section $4$, we give two conditions that a constacyclic code is GRS. Section~5 concludes our work.

\section{Preliminaries}
Let $\mathbb{F}_{q}$ be the finite field of $q$ elements, where $q$ is a prime power. Let $\mathbb{F}_{q}^{*}=\mathbb{F}_{q}\backslash\{0\}$ be the multiplicative group of $\mathbb{F}_q$. For two integers $c,~d$, the greatest common divisor of $c$ and $d$ is denoted by $(c,~d)$. The size of a finite set $S$ is denoted by $|S|$. For a positive integer $m$, $\mathbb{Z}_{m}$ denotes the ring of integers modulo $m$,~$\mathbb{Z}_{m}^{\times}$ denotes the set of all invertible elements of $\mathbb{Z}_{m}$ with respect to multiplication.

For an abelian group $G$, we denote the operation of $G$ as $+$. For nonempty subsets $A_{1},~A_{2},~A$ of $G$, define
$$A_{1}+A_{2}=\{a_{1}+a_{2}\mid a_{1}\in A_{1},~a_{2}\in A_{2} \},$$
$$-A=\{-a\mid a\in A\},$$
$$A_{1}-A_{2}=A_{1}+(-A_{2})$$
and
$$A^{c}=G\setminus A=\{g \mid g\in G,~g\notin A\}.$$

For a linear $[\,n, k, d\,]_{q}$ code $\mathcal{C}$ and a codeword $\bm c\in\mathcal{C}$, we denote the Hamming weight of $\bm c$ as $wt(\bm c)$. The dual code $\mathcal{C}^{\perp}$ of $\mathcal{C}$ is defined by
$$\mathcal{C}^{\perp}=\{(b_{0},\dots,b_{n-1})\in\mathbb{F}_{q}^{n}\mid\sum_{i=0}^{n-1}b_{i}c_{i}=0,~\forall~(c_{0},\dots,c_{n-1})\in\mathcal{C}\}.
$$

Let $\mathbb{F}_{q}[x]$ be the polynomial ring over $\mathbb{F}_{q}$.
Let $\bm\alpha=(\alpha_{1},\alpha_{2},...,\alpha_{n})$ and $\bm v=(v_1,v_2,\dots, v_n)$ be two vectors of length $n$ over $\mathbb{F}_q$. We define the evaluation map related to $\bm\alpha$ and $\bm v$ as follows:

$$
ev_{\bm\alpha, \bm v}:\mathbb{F}_{q}[x]\rightarrow\mathbb{F}_{q}^{n},~f(x) \mapsto (v_1f(\alpha_{1}),v_2f(\alpha_{2}),...,v_nf(\alpha_{n})).
$$
\begin{Definition}
Let $\alpha_{1},\dots,\alpha_{n}\in\mathbb{F}_{q}\bigcup\{\infty\}$ be distinct elements, $k<n$, and $v_{1},\dots,v_{n}\in\mathbb{F}_{q}^{*}$. The corresponding generalized Reed-Solomon (GRS) code is defined by
$$
GRS_{n,k} = \{(v_{1}f(\alpha_{1}),\dots,v_{n}f(\alpha_{n}))\mid f\in \mathbb{F}_{q}[x], \deg f<k\}.
$$
In this setting, for a polynomial $f(x)$ of degree $\deg f(x)<k$, the quantity $f(\infty)$ is defined as the coefficient of $x^{k-1}$ in the polynomial $f(x)$. In the case $v_{i}=1$ for all $i$, the code is called a Reed-Solomon (RS) code.
\end{Definition}

It is well-known that GRS codes are MDS codes and the dual of GRS codes are also GRS.

\begin{Definition}\label{def-1}
Let $\bm x=(x_{1},\dots,x_{n})$, $\bm y=(y_{1},\dots,y_{n})\in\mathbb{F}_{q}^{n}$, the {\it Schur product} of $\bm x$ and $\bm y$ is defined as $\bm x*\bm y=(x_{1}y_{1},\dots,x_{n}y_{n})$. The product of two linear codes $\mathcal{C}_{1}$, $\mathcal{C}_{2}\subseteq\mathbb{F}_{q}^{n}$ is defined as
$$\mathcal{C}_{1}*\mathcal{C}_{2}=\left\langle\bm c_{1}*\bm c_{2}\mid\bm c_{1}\in\mathcal{C}_{1},~\bm c_{2}\in\mathcal{C}_{2}\right\rangle_{\mathbb{F}_{q}}$$ where $\left\langle S\right\rangle_{\mathbb{F}_{q}}$ denotes the $\mathbb{F}_{q}$-linear subspace generated by the subset $S$ of $\mathbb{F}_{q}^{n}$.

In particular, if $\mathcal{C}_{1}=\mathcal{C}_{2}$, we call $\mathcal{C}^{2}=\mathcal{C}*\mathcal{C}$ the square code.
\end{Definition}

In this paper, we always assume $(q,~n)=1$. For $\lambda\in\mathbb{F}_{q}^{*}$, a $q-$ary $[n,~k]$ linear code $\mathcal{C}$ is called $\lambda-$constacyclic if $(c_{0},~c_{1},\dots,c_{n-1})\in\mathcal{C}$ implies $(\lambda c_{n-1},~c_{0},\dots,c_{n-2})\in\mathcal{C}$.
There is an $\mathbb{F}_{q}-$linear isomorphism between $\mathbb{F}_{q}^{n}$ and $R_{n,~\lambda}=\mathbb{F}_{q}[x]/(x^{n}-\lambda)$ which is defined by
$$(c_{0},~c_{1},\dots,c_{n-1})\mapsto c_{0}+c_{1}x+\dots+c_{n-1}x^{n-1}+(x^{n}-\lambda)
$$
where $(x^{n}-\lambda)$ is the ideal of $\mathbb{F}_{q}[x]$ generated by $x^{n}-\lambda$. Then a $\lambda-$constacyclic code of length $n$ over $\mathbb{F}_{q}$ is an ideal of the ring $R_{n,~\lambda}=\mathbb{F}_{q}[x]/(x^{n}-\lambda)$. Since $R_{n,~\lambda}$ is a principal ideal ring, every $\lambda-$constacyclic code $\mathcal{C}$ is generated uniquely by a monic divisor $g(x)$ of $x^{n}-\lambda$ and denoted by $\mathcal{C}=(g(x))$. We call $g(x)$ and $h(x)=\frac{x^{n}-\lambda}{g(x)}=\sum_{i=0}^{k}h_{i}x^{i}$ the generator polynomial and the check polynomial of $\mathcal{C}$, respectively. The dimension of $\mathcal{C}$ equals to $n-\deg g(x)$.

Suppose the splitting field of $x^{n}-\lambda$ over $\mathbb{F}_{q}$ is $\mathbb{F}_{q^{d}}$. Then, there exist $\beta\in\mathbb{F}_{q^{d}}$ and a primitive $n$th root of unity $\alpha$ such that $\{\beta\alpha^{i}\mid 0\leq i\leq n-1\}$ is the set of the roots of $x^{n}-\lambda$.

\begin{Definition}\label{def-zero}
With the notations above, let $\mathcal{C}$ be a $\lambda-$constacyclic code with generator polynomial $g(x)$. Suppose the set of roots of $g(x)$ is $\{\beta\alpha^{i_{j}}\mid 0\leq j\leq n-k\}$, we call the set $I=\{i_{j}\mid 1\leq j\leq n-k\}$ the defining set of $\mathcal{C}$.
\end{Definition}

Obviously, the dimension of $\mathcal{C}$ equals to $n-|I|$. Since $\alpha$ is a primitive $n$th root of unity in $\mathbb{F}_{q^{d}}$, we can view $I$ as a subset of $\mathbb{Z}_{n}$. It is well-known that the dual code $\mathcal{C}^{\perp}$ is a $\lambda^{-1}-$constacyclic code generated by $$g^{\perp}(x)=h_{0}^{-1}x^{k}h(x^{-1})=\sum_{i=0}^{k}h_{i}h_{0}^{-1}x^{k-i},$$
 then the defining set of $\mathcal{C}^{\perp}$ is
 $-I^{c}=-(\mathbb{Z}_{n}\setminus I).$

 Thus,~
 $\{\beta^{-1}\alpha^{i}\mid 0\leq i\leq n-1\}$
  is the set of the roots of $x^{n}-\lambda^{-1}$.
  And
  $\{\beta^{-1}\alpha^{l}\mid l\in -I^{c}\}$
   is the set of the roots of the generator polynomial of $\mathcal{C}^{\perp}$.

\section{Schur square of constacyclic codes}
In this section, we study the square of constacyclic codes. In \cite{I. Cascudo}, the author characterized the square of cyclic codes based on the results of \cite{J. Bierbrauer}. In this section, we slightly extend the result to constacyclic codes.

Let $\mathcal{C}$ be a $\lambda-$constacyclic code of length $n$ and dimension $k$ over $\mathbb{F}_{q}$ with generator polynomial $g(x)$. Suppose $I$ is the defining set of $\mathcal{C}$, then $|I|=n-k$. Recall that $\mathbb{F}_{q^{d}}$ is the splitting field of $x^{n}-\lambda$ over $\mathbb{F}_{q}$.

Let $\widetilde{\mathcal{C}}(I)$ be an $\mathbb{F}_{q^{d}}-$linear code generated by
$$
G_{I}=\left[
\begin{matrix}
1&\beta\alpha^{i_{1}}&\dots&(\beta\alpha^{i_{1}})^{n-1}\\
1&\beta\alpha^{i_{2}}&\dots&(\beta\alpha^{i_{2}})^{n-1}\\
\vdots&\vdots&&\vdots\\
1&\beta\alpha^{i_{n-k}}&\dots&(\beta\alpha^{i_{n-k}})^{n-1}
\end{matrix}
\right]_{(n-k)\times n}.
$$
where $I=\{i_{1},\dots,i_{n-k}\}$.

Suppose $-I^{c}=\{j_{1},\dots,j_{k}\}$, we know that the $\mathbb{F}_{q^{d}}-$dual code $\widetilde{\mathcal{C}}(I)^{\perp}$ of $\widetilde{\mathcal{C}}(I)$ is generated by
$$
H_{-I^{c}}=\left[
\begin{matrix}
1&\beta^{-1}\alpha^{j_{1}}&\dots&(\beta^{-1}\alpha^{j_{1}})^{n-1}\\
1&\beta^{-1}\alpha^{j_{2}}&\dots&(\beta^{-1}\alpha^{j_{2}})^{n-1}\\
\vdots&\vdots&&\vdots\\
1&\beta^{-1}\alpha^{j_{k}}&\dots&(\beta^{-1}\alpha^{j_{k}})^{n-1}
\end{matrix}
\right]_{k\times n}.
$$

Let $\mathbb{F}\supseteq\mathbb{F}_{q}$ be a field extension and $\mathcal{A}$ be an linear code of length $n$ over $\mathbb{F}$, we denote $\mathcal{A}|_{\mathbb{F}_{q}}=\mathcal{A}\bigcap\mathbb{F}_{q}^{n}$ the subfield subcode of $\mathcal{A}$. Since a codeword $\bm c=(c_{0},\dots,c_{n-1})\in\mathcal{C}$ if and only if $c(\beta\alpha^{l})=0$ for any $l\in I$, where $c(x)=\sum_{i=0}^{n-1}c_{i}x^{i}$,~then $\mathcal{C}=\widetilde{\mathcal{C}}(I)^{\perp}|_{\mathbb{F}_{q}}$.

For an $h$ dimensional linear code $\mathcal{B}\in\mathbb{F}_{q}^{n}$, suppose $\{\bm b_{1},\dots,\bm b_{h}\}$ is a base of the code $\mathcal{B}$, then we define an $\mathbb{F}-$linear code as
 $$\mathbb{F}\bigotimes_{\mathbb{F}_{q}}\mathcal{B}=\{a_{1}\bm b_{1}+\dots+a_{h}\bm b_{h}\mid a_{i}\in \mathbb{F}\}.$$

\begin{Proposition}\label{pro-con}
Let $I$ be the defining set of the $\lambda-$constacyclic code $\mathcal{C}$. Then $\mathcal{C}^{2}$ is a $\lambda^{2}-$constacyclic code with defining set $(I^{c}+I^{c})^{c}$. In particular, the dimension of $\mathcal{C}^{2}$ equals to $|I^{c}+I^{c}|$.
\end{Proposition}
\begin{proof}
Since $\dim_{\mathbb{F}_{q}}(\mathcal{C})=\dim_{\mathbb{F}_{q^{d}}}(\widetilde{\mathcal{C}}(I)^{\perp})$ and $\mathcal{C}\subseteq\widetilde{\mathcal{C}}(I)^{\perp}$, we know that $\widetilde{\mathcal{C}}(I)^{\perp}=\mathbb{F}_{q^{d}}\bigotimes_{\mathbb{F}_{q}}\mathcal{C}$. By Lemma 2.23 of \cite{H. Randriambololona}, $$\mathbb{F}_{q^{d}}\bigotimes_{\mathbb{F}_{q}}\mathcal{C}^{2}=(\mathbb{F}_{q^{d}}\bigotimes_{\mathbb{F}_{q}}\mathcal{C})^{2}.$$

Let $\bm v_{s}$ denote the $s$th row of the matrix $H_{-I^{c}}$, where $1\leq s\leq k$. By Definition~\ref{def-1}, we know that
$$
(\mathbb{F}_{q^{d}}\bigotimes_{\mathbb{F}_{q}}\mathcal{C})^{2}=\left\langle \bm v_{i}*\bm v_{j}\mid 1\leq i,j \leq k \right\rangle_{\mathbb{F}_{q^{d}}}.
$$
Thus, $(\mathbb{F}_{q^{d}}\bigotimes_{\mathbb{F}_{q}}\mathcal{C})^{2}$ is generated by
$$
\left[
\begin{matrix}
1&\beta^{-2}\alpha^{l_{1}}&\dots&(\beta^{-2}\alpha^{l_{1}})^{n-1}\\
1&\beta^{-2}\alpha^{l_{2}}&\dots&(\beta^{-2}\alpha^{l_{2}})^{n-1}\\
\vdots&\vdots&&\vdots\\
1&\beta^{-2}\alpha^{l_{t}}&\dots&(\beta^{-2}\alpha^{l_{t}})^{n-1}
\end{matrix}
\right]_{t\times n},
$$
where $-(I^{c}+I^{c})=\{l_{1},\dots,l_{t}\}$.

Hence, we have $\mathcal{C}^{2}=\mathbb{F}_{q^{d}}\bigotimes_{\mathbb{F}_{q}}\mathcal{C}^{2}|_{\mathbb{F}_{q}}=
(\mathbb{F}_{q^{d}}\bigotimes_{\mathbb{F}_{q}}\mathcal{C})^{2}|_{\mathbb{F}_{q}}$, which completes the proof.
\end{proof}


\begin{Remark}
In \cite{I. Cascudo}, the author studied square of cyclic codes. Theorem 3.3 of \cite{I. Cascudo} is the case $\lambda=1$ in Proposition~\ref{pro-con}.
\end{Remark}

In the following, we give three examples of the codes constructed by M. Grassl and M. R\"{o}tteler in \cite{M. Grassl}. In \cite{S. Ball}, the author proved that these three MDS constacyclic codes are GRS codes by a technical method. We first characterize them by Proposition~\ref{pro-con}. We will also use these examples in Section 4.

\begin{Example}\label{exm-3}
Let $\omega$ be a primitive element of $\mathbb{F}_{q^{2}}$ and let $\alpha=\omega^{q-1}$ be a primitive $(q+1)$th root of unity. Let $\mathcal{C}_{1}=(g_{1}(x))$ be an $\omega^{q+1}-$constacyclic code where $g_{1}(x)|(x^{q+1}-\omega^{q+1})$. For $\dim(\mathcal{C}_{1})$ even and $q$ odd, let $g_{1}(x)=\prod_{i=-s+1}^{s}(x-\omega\alpha^{i})$ then $\mathcal{C}_{1}=(g_{1}(x))$ is an MDS code of dimension $q-2s+1$. The defining set of $\mathcal{C}_{1}$ is $I_{1}=\{-s+1,~-s+2,\dots,s\}\subseteq\mathbb{Z}_{q+1}$. By Proposition~\ref{pro-con}, $\mathcal{C}_{1}^{2}$ is an ~$\omega^{2q+2}-$constacyclic code with defining set $(I_{1}^{c}+I_{1}^{c})^{c}$.
\end{Example}

\begin{Example}\label{exm-1}
Let $\omega$ be a primitive element of $\mathbb{F}_{q^{2}}$ and let $\alpha=\omega^{q-1}$ be a primitive $(q+1)$th root of unity. Let $\mathcal{C}_{2}=(g_{2}(x))$ be a cyclic code where $g_{2}(x)|(x^{q+1}-1)$. For $q$ and $\dim(\mathcal{C}_{2})$ both odd or both even, let $g_{2}(x)=\prod_{i=-s}^{s}(x-\alpha^{i})$ then $\mathcal{C}_{2}=(g_{2}(x))$ is an MDS cyclic code of dimension $q-2s$. The defining set of $\mathcal{C}_{2}$ is $I_{2}=\{-s,-s+1,\dots,s\}\subseteq\mathbb{Z}_{q+1}$. By Proposition~\ref{pro-con}, $\mathcal{C}_{2}^{2}$ is a cyclic code with defining set $(I_{2}^{c}+I_{2}^{c})^{c}$.
\end{Example}

\begin{Example}\label{exm-2}
Let $\omega$ be a primitive element of $\mathbb{F}_{q^{2}}$ and let $\alpha=\omega^{q-1}$ be a primitive $(q+1)$th root of unity. Let $\mathcal{C}_{3}=(g_{3}(x))$ be a cyclic code where $g_{3}(x)|(x^{q+1}-1)$. For $\dim(\mathcal{C}_{3})$ odd and $q$ even, let $g_{3}(x)=\prod_{i=\frac{1}{2}q-s}^{\frac{1}{2}q+s+1}(x-\alpha^{i})$ then $\mathcal{C}_{3}=(g_{3}(x))$ is an MDS cyclic code of dimension $q-2s-1$. The defining set of $\mathcal{C}_{3}$ is $I_{3}=\{\frac{1}{2}q-s,\dots,\frac{1}{2}q+s+1\}\subseteq\mathbb{Z}_{q+1}$
 By Proposition~\ref{pro-con}, $\mathcal{C}_{3}^{2}$ is a cyclic code with defining set $(I_{3}^{c}+I_{3}^{c})^{c}$.
\end{Example}

\section{Conditions for a class of MDS constacyclic codes to be RS codes }
In this section, we give two conditions that a constacyclic code is GRS. First, Corollary $2$ in \cite{P. BeelenS} mentioned that an $[\,n, k\,]$ MDS code with $k\leq 2$ or $n-k\leq 2$ is a GRS code. Hence, in the following, we assume that $k>2$ or $n-k>2$.

We present two useful well-known results.

\begin{Lemma}(Cauchy~\cite{A. Cauchy}-Davenport~\cite{H. Davenport})\label{the-1}
If $A,~B\subseteq\mathbb{Z}_{p}$ are nonempty and $p$ is prime, then
$$|A+B|\geq \min\{p,~|A|+|B|-1\}.
$$
\end{Lemma}

A subset $D$ of $\mathbb{Z}_{n}$ is called an {\it arithmetic progression} of difference $g\in\mathbb{Z}_{n}$ with $g\neq0$, if there exist $\beta\in\mathbb{Z}_{n}$ such that $D-\beta=\{ig\mid0\leq i\leq|D|-1\}$. If two nonempty subsets $A,~B$ of $\mathbb{Z}_{n}$ satisfy the equality in Lemma~\ref{the-1}, then the pair $(A,~B)$ is called a {\it critical pair}.

\begin{Lemma}(Vosper~\cite{A. G. Vosper})\label{the-Vosp}
(A,~B) is a critical pair of nonempty subsets of $\mathbb{Z}_{p}$ where $p$ is prime, if and only if one of the following holds.

(1) $|A|+|B|>p$ and $A+B=\mathbb{Z}_{p}$.

(2) $|A|+|B|=p$ and $|A+B|=p-1$.

(3) $\min\{|A|,|B|\}=1$.

(4) $A$ and $B$ are arithmetic progressions with a common difference.
\end{Lemma}

The following proposition is crucial for our results in this section.

\begin{Proposition}\cite{D. Mirandola}\label{pro-zy}
Let $\mathcal{C}\subseteq\mathbb{F}_{q}^{n}$ be an MDS code with $\dim(C)\leq\frac{n-1}{2}$. The code $\mathcal{C}$ is GRS if and only if $\dim(\mathcal{C}^{2})=2\dim(\mathcal{C})-1$.
\end{Proposition}

In \cite{D. Mirandola}, when $\dim(\mathcal{C})=\frac{n}{2}$, the authors showed that we could not determine whether $\mathcal{C}$ is GRS by $\dim(\mathcal{C}^{2})$.

\begin{Theorem}\label{the-cf}
Let $\mathcal{C}$ be a $\lambda-$constacyclic code of length $n$ and dimension $k$ over $\mathbb{F}_{q}$ with defining set $I$, and suppose $2<k\leq\frac{n-1}{2}$ or $\frac{n+1}{2}\leq k<n-2$. If $I$ is an arithmetic progression with the difference $a\in\mathbb{Z}_{n}^{\times}$, then $\mathcal{C}$ is a GRS code.
\end{Theorem}
\begin{proof}
Since $I$ is an arithmetic progression with the difference $a\in\mathbb{Z}_{n}^{\times}$, then there exists an element $b\in \mathbb{Z}_{n}$ such that $$I-b=\{ai\mid i=0,1,\dots,n-k-1\}.$$ Obviously,
\begin{equation*}
\begin{aligned}
I^{c}&=b+\{aj\mid j=n-k,~n-k+1,\dots,n-1\}\\
     &=b+a(n-k)+\{at\mid t=0,1,\dots,k-1\}.
\end{aligned}
\end{equation*}
It is easy to see that $I^{c}$ and $-I^{c}$
are also arithmetic progressions with the difference $a$.

Suppose
$$
G_{I}=\left[
\begin{matrix}
1&\beta\alpha^{i_{1}}&\dots&(\beta\alpha^{i_{1}})^{n-1}\\
1&\beta\alpha^{i_{2}}&\dots&(\beta\alpha^{i_{2}})^{n-1}\\
\vdots&\vdots&&\vdots\\
1&\beta\alpha^{i_{n-k}}&\dots&(\beta\alpha^{i_{n-k}})^{n-1}
\end{matrix}
\right]_{(n-k)\times n}.
$$
where $I=\{i_{1},\dots,i_{n-k}\}$.

Since $(a,n)=1$, then $\alpha^{a}$ is also a primitive $n$th root of unity. Hence, it is easy to see that any $n-k$ columns of $G_{I}$ are linearly independent. For any $\bm c=(c_{0},c_{1},\dots,c_{n-1})\in\mathcal{C}$, we have $G_{I}\bm c^{T}=\bm 0$, thus $wt(\bm c)\geq n-k+1$, then we have $\mathcal{C}$ is MDS.

When $2<k\leq\frac{n-1}{2}$, by Proposition~\ref{pro-con}, we know that the defining set of $\mathcal{C}^{2}$ is
\begin{equation*}
\begin{aligned}
(I^{c}+I^{c})^{c}&=(2b+2a(n-k)+\{at\mid t=0,1,\dots,2k-2\})^{c}\\
&=2b+2a(n-k)+\{at\mid t=2k-1,2k\dots,n-1\}.
\end{aligned}
\end{equation*}
Thus the dimension of $\mathcal{C}^{2}$ equals to $n-|(I^{c}+I^{c})^{c}|=2k-1$. By Proposition~\ref{pro-zy}, we know that $\mathcal{C}$ is GRS.

Since the dual code of a GRS code is also a GRS code, and the defining set of $\mathcal{C}^{\perp}$ is $-I^{c}$, then the same argument can be applied for the case $\frac{n+1}{2}\leq k<n-2$.
\end{proof}

\begin{Theorem}\label{the-cy}
Let $\mathcal{C}$ be a $\lambda-$constacyclic code of length $n=p$ and dimension $k$ over $\mathbb{F}_{q}$ where $p$ is an odd prime and $(p,q)=1$. Suppose the defining set of $\mathcal{C}$ is $I$ and $2<k\leq\frac{n-1}{2}$ or $\frac{n+1}{2}\leq k<n-2$, then $\mathcal{C}$ is a GRS code if and only if $I$ is an arithmetic progression with difference $a$ for some $a\neq0$.

\end{Theorem}
\begin{proof}
Note that $a\neq0$ iff $(a,p)=1$ in $\mathbb{Z}_{p}$, then $\alpha^{a}$ is also a primitive $n$th root of unity.
The ``if" part can be obtained directly from Theorem~\ref{the-cf}, so we only need to prove the ``only if" part.

When $2<k\leq\frac{n-1}{2}$. By Proposition~\ref{pro-zy}, the dimension of $\mathcal{C}^{2}$ equals to $2k-1$. Combining with Proposition~\ref{pro-con}, we know that $(I^{c},I^{c})$ is a critical pair. Then by Lemma~\ref{the-Vosp}, we know that $I^{c}$ is an arithmetic progression which implies $I$ is also an arithmetic progression.

Since the dual code of a GRS code is also a GRS code, and the defining set of $\mathcal{C}^{\perp}$ is $-I^{c}$, then the same argument can be applied for the case $\frac{n+1}{2}\leq k<n-2$.
\end{proof}

\begin{Remark}
Since $n=p$ is an odd prime, then $\frac{n}{2}$ is not an integer. Hence, Theorem~\ref{the-cy} holds for all $k$.
\end{Remark}

\begin{Remark}\label{rem-cyc2}
When $\lambda=1$,~$\mathcal{C}$ is cyclic, then the same results as corollaries of Theorems~\ref{the-cf},~\ref{the-cy} hold for cyclic codes.
\end{Remark}

\begin{Remark}\label{rem-cyc3}
In \cite{S. Ball}, the author proved that the codes constructed by M. Grassl and M. R\"{o}tteler in \cite{M. Grassl} (see the examples in Section 3) are GRS codes. Obviously, the defining set of three classes MDS constacyclic codes in \cite{M. Grassl} are arithmetic progressions, then by the results in this section, we also get the same results as in \cite{S. Ball} except for $k=\frac{n}{2}$.
\end{Remark}

\begin{Remark}\label{rem-cyc5}
The defining sets of the MDS constacyclic codes in \cite{X. Kai}, \cite{T. Zhang} and \cite{B. Chen} are arithmetic progressions with invertible differences, then we know that these codes are GRS codes when $k\neq\frac{n}{2}$.
\end{Remark}

\section{Conclusion}
 This paper investigates under what conditions that MDS constacyclic codes are GRS. We study the square of constacyclic codes by defining set. Then, we give a sufficient condition that a constacyclic code is GRS. In particular, We provide a necessary and sufficient condition that a constacyclic code of prime length is GRS. By our results, we get the same result as in \cite{S. Ball} except for $k=\frac{n}{2}$. We also show that the MDS constacyclic codes in \cite{X. Kai}, \cite{T. Zhang} and \cite{B. Chen} are GRS codes when $k\neq\frac{n}{2}$.

\vskip 4mm

\noindent {\bf Acknowledgement.} This work was supported by NSFC (Grant No. 11871025, 12271199).


\end{document}